\documentclass[a4paper,USenglish,cleveref,autoref,thm-restate]{lipics-v2021}

\pdfoutput=1
\hideLIPIcs
\nolinenumbers

\title{A weakly monotonic, logically constrained, HORPO-variant}

\author{Cynthia Kop}{Radboud University, Nijmegen, The Netherlands}{c.kop@cs.ru.nl}{https://orcid.org/0000-0002-6337-2544}{}

\authorrunning{L. Guo and C. Kop}

\Copyright{Liye Guo and Cynthia Kop}

\ccsdesc[500]{Theory of computation~Equational logic and rewriting}

\keywords{Higher-order term rewriting, constrained rewriting, dependency pairs}

\funding{The author is supported by NWO VI.Vidi.193.075, project ``CHORPE''.}

\usepackage{mathtools}

\newcommand{\Var}{\mathit{Var}}
\newcommand{\gtvA}{L}
\newcommand{\downarrowcalc}{\downarrow_\kappa}
\newcommand{\arrz}{\rightarrow}
\newcommand{\arr}[1]{\arrz_{#1}}
\newcommand{\arrcalc}{\arr{\kappa}}
\newcommand{\arrtype}{\Rightarrow}

\newcommand{\grmain}{\sqsupset}
\newcommand{\geqmain}{\sqsupseteq}
\newcommand{\ismain}{\approx}
\newcommand{\rpomain}{\sqsupset\!\!\sqsupset}

\newcommand{\geqth}{\succeq_\varphi^\gtvA}
\newcommand{\grth}{\succ_\varphi^\gtvA}
\newcommand{\rpoth}{\succ\!\!\!\succ_\varphi^\gtvA}
\newcommand{\rpo}{\succ\!\!\!\succ}

\newcommand{\grpred}{\rhd}
\newcommand{\geqpred}{\unrhd}
\newcommand{\leqpred}{\unlhd}
\newcommand{\eqpred}{\equiv}

\newcommand{\grsort}{\blacktriangleright}
\newcommand{\geqsort}{\unrhd\!\!\!\!\!\!\blacktriangleright}

\newcommand{\symb}[1]{\mathtt{#1}}
\newcommand{\afun}{\symb{f}}
\newcommand{\bfun}{\symb{g}}
\newcommand{\cfun}{\symb{h}}
\newcommand{\atype}{\sigma}
\newcommand{\btype}{\tau}
\newcommand{\asort}{\iota}
\newcommand{\bsort}{\iota'}
\newcommand{\avar}{x}
\newcommand{\bvar}{y}
\newcommand{\cvar}{z}

\newcommand{\filter}{\pi}

\begin{document}

\maketitle

\begin{abstract}
In this short paper, we present a simple variant of the recursive path ordering, specified for
Logically Constrained Simply Typed Rewriting Systems (LCSTRSs).  This is a method for
\emph{curried} systems, without $\lambda$ but with partially applied function symbols, which can
deal with \emph{logical constraints}.  As it is designed for use in the dependency pair framework,
it is defined as reduction pair, allowing weak monotonicity.
\end{abstract}

\section{Introduction}

This is a technical report, so no preliminaries are given.
We assume familiarity with Logically Constrained Simply Typed Rewriting
Systems~\cite{guo:kop:24}, and will adapt the HORPO variant given in that paper to a reduction
pair (although for now, we omit the multiset status as it would complicate the definition).

The main difference with~\cite{guo:kop:24} is a kind of filtering to take advantage of the weaker
monotonicity requirements that are used for reduction pairs in the DP framework.  This filtering is
reminiscent of \emph{argument filterings} which are often used in dependency pair approaches, but
note that the direct definition of argument filtering does not naturally extend to systems with
partially applied function symbols -- so a special definition is needed.

We will use the following definigion of a constrained reduction pair:

\begin{definition}\label{def:redpair}
A \emph{constrained relation} is a set R of tuples $(s,t,\varphi,\gtvA)$. denoted $s\ 
R_\varphi^\gtvA\ t$, where $s$ and $t$ are terms of the same type, $\varphi$ is a constraint, and
$\gtvA$ is a set of variables.

We say a binary relation $R'$ on terms \emph{covers} $R$ if $s\ R_\varphi^\gtvA\ t$ implies that
$(s\gamma)\downarrowcalc\ R'\ (t\gamma\downarrowcalc)$ for any substitution $\gamma$ that respects
$\varphi$ and maps all $x \in \gtvA$ to ground theory terms.

A \emph{constrained reduction pair} is a pair $(\succeq,\succ)$ of constrained relations such that
there exist
a \emph{reflexive} relation $\geqmain$ that covers $\succeq$ and
a \emph{well-founded} relation $\grmain$ that covers $\succ$
such that:
\begin{itemize}
\item $\arrcalc\, \subseteq\, \geqmain$, and
\item $\geqmain \cdot \grmain\, \subseteq\, \grmain^+$ or
  $\grmain \cdot \geqmain\, \subseteq\, \grmain^+$, and
\item $\geqmain$ is monotonic: $s \geqmain t$ implies $C[s] \sqsupseteq C[t]$
  for every appropriately-typed context $C$.
\end{itemize}
If also $\grmain$ is monotonic, we say $(\succeq,\succ)$ is a \emph{strongly monotonic
constrained reduction pair}.
\end{definition}

As most existing HORPO variants actually consider systems where function symbols are maximally
applied, we will start by defining a variation for (unconstrained) STRS.
Then, in Section~\ref{sec:constrained}, we will use this variant as $(\geqmain,\grmain)$ in
Definition~\ref{def:redpair}.

Notationally, we will use the following conventions without expicitly stating them:
\begin{itemize}
\item $\asort$ always refers to a sort (base type); $\atype,\btype$ to arbitrary types
\item $\afun,\bfun,\cfun$ always refers to a function symbol, $\avar,\bvar,\cvar$ to a variable
\end{itemize}

\section{Unconstrained HORPO for curried systems}\label{sec:unconstrained}

We assume familiarity with STRSs: in short, terms are constructed from a set of typed variables,
typed constants, and the type-conscious application operator (these are LCSTRSs with an empty
theory signature).
We allow for an infinite set of function symbols (constants), but assume given:
\begin{itemize}
\item a \emph{precedence}: a quasi-ordering $\geqpred$ on the set of all function symbols, such
  that its strict part $\grpred ::= \geqpred \setminus \leqpred$, is well-founded;
  we write $\eqpred$ for the equivalence relation $\geqpred \cap \leqpred$  
\item a \emph{filter}: a function $\filter$ mapping each function symbol $\afun :: \atype_1
  \arrtype \dots \arrtype \atype_m \arrtype \asort$ with $\asort$ a sort to a subset of
  $\{1,\dots,m\}$: the arguments that $\afun$ regards
\end{itemize}
Moreover, while we allow an infinite number of symbols with the same status, we require that their
maximal arity (that is, the number $m$ if $\afun :: \atype_1 \arrtype \dots \arrtype \atype_m
\arrtype \asort$) is bounded.

\medskip
In all the following, we collapse base types; that is, we assume that there is only one sort
$\asort$.  Essentially, this means that we consider (and will do induction on) type \emph{structure}
rather than directly on types.
We will define a number of relations.

\subsection{Equivalence}
The first of our relations can be defined without reference to the others.
For two terms $s,t$, we define $s \ismain t$ if $s$ and $t$ have the same type structure, and one of
the following holds:
\begin{description}
\item[(Eq-mono)] $s = \avar\ s_1 \cdots s_n$ and $t = \avar\ t_1 \cdots t_n$ for $\avar$ a variable,
  and each $s_i \ismain t_i$
\item[(Eq-args)] $s = \afun\ s_1 \cdots s_n$ and $t = \bfun\ t_1 \cdots t_n$ for $\afun,\bfun$
  function symbols with the same type structure such that
  $\afun \eqpred \bfun$,
  $\filter(\afun) = \filter(\bfun)$, and 
  $s_i \ismain t_i$ for all $i \in \filter(\afun) \cap \{1,\dots,n\}$.
\end{description}

We observe that $\ismain$ is a monotonic and stable equivalence relation.

\begin{lemma}\label{lem:ismain}
The relation $\ismain$ is transitive, reflexive, symmetric, monotonic, and stable.
\end{lemma}

\begin{proof}
By a straighforward induction on the size of a given term $s$ we can prove:

\emph{Transitivity:}
if $s \ismain t \ismain u$ then $s \ismain u$.

\emph{Reflexivity:} $s \ismain s$

\emph{Symmetry:} if $s \ismain t$ then $t \ismain s$

\emph{Stability:} if $s \ismain t$ and $\gamma$ is a substitution, then $s\gamma \ismain t\gamma$;
  for the case $s = x\ s_1 \cdots s_n$ we do a case analysis on $\gamma(x)$ (which must either have
  a form $y\ w_1 \cdots w_k$ or $\afun\ w_1 \cdots w_k$), and the reflexivity property we have
  shown above (to find $w_i \ismain w_i$)

For monotonicity, we do not need the induction.  Suppose $s :: \atype \arrtype \btype$ and $t ::
\atype$, with $s \ismain s'$ and $t \ismain t'$.  There are two cases:
\begin{itemize}
\item $s = x\ s_1 \cdots s_n$ and $s' = x\ s_1' \cdots s_n'$ with each $s_i \ismain s_i'$; since
  $t \ismain t'$ also $x\ s_1 \cdots s_n\ t \ismain x\ s_1' \cdots s_n'\ t'$
\item $s = \afun\ s_1 \cdots s_n$ and $s' = \bfun\ s_1' \cdots s_n'$ with $\afun,\bfun$ having the
  same type structure and filter, and being equal in the precedence; this is not affected by adding
  an extra argument, so we complete as above.
  \qedhere
\end{itemize}
\end{proof}

\subsection{Decrease}

Our other relations are defined through a mutual recursion.  Specifically, for terms $s,t$:

\begin{itemize}
\item $s \geqmain t$ if $s \ismain t$ or $s \grmain t$
\item $s \grmain t$ if $s$ and $t$ have the same type structure, and:
  \begin{description}
  \item[(Gr-mono)] $s = x\ s_1 \cdots s_n$ and $t = x\ t_1 \cdots t_n$ and each $s_i \geqmain t_i$
    and some $s_i \grmain t_i$, or
  \item[(Gr-args)] $s = \afun\ s_1 \cdots s_n$ and $t = \bfun\ t_1 \cdots t_n$ for $\afun,\bfun$
    function symbols with the same type structure such that $\afun \eqpred \bfun$,
    $\filter(\afun) = \filter(\bfun)$, and 
    $s_i \geqmain t_i$ for all $i \in \filter(\afun) \cap \{1,\dots,n\}$, and
    $s_i \grmain t_i$ for some $i \in \filter(\afun) \cap \{1,\dots,n\}$
  \item[(Gr-rpo)] $s \rpomain t$
  \end{description}
\item $s \rpomain t$ if $s = \afun\ s_1 \cdots s_n$ with $\afun :: \atype_1 \arrtype \dots
  \arrtype \atype_m \arrtype \asort$ and $\{n+1,\dots,m\} \subseteq \filter(\afun)$, and:
  \begin{description}
  \item[(Rpo-select)] $s_i \geqmain t$ for some $i \in \filter(\afun) \cap \{1,\dots,n\}$
  \item[(Rpo-appl)] $t = t_0\ t_1 \cdots t_m$ and $s \rpomain t_i$ for $1 \leq i \leq m$
  \item[(Rpo-copy)] $t = \bfun\ t_1 \cdots t_m$ with $\afun \grpred \bfun$ and 
    $s \rpomain t_i$ for all $i \in \filter(\bfun) \cap \{1,\dots,m\}$
  \item[(Rpo-lex)] $t = \bfun\ t_1 \cdots t_m$ with $\afun \eqpred \bfun$ and there is some
    $i \in \filter(\afun) \cap \filter(\bfun) \cap \{1,\dots,\min(n,m)\}$
    such that all of the following hold:
    \begin{itemize}
    \item $\filter(\afun) \cap \{1,\dots,i\} = \filter(\bfun) \cap \{1,\dots,i\}$
    \item $s_j \ismain t_j$ for $j \in \{1,\dots,i-1\} \cap \filter(\afun)$
    \item $s_i \grmain t_i$
    \item $s \rpomain t_j$ for $j \in \{i+1,\dots,m\} \cap \filter(\bfun)$
    \end{itemize}
  \end{description}
\end{itemize}

\subsection{Properties}

We make several observations:

\begin{lemma}\label{lem:compatibility}
If $s \ismain t\ R\ u$ for $R \in \{\ismain,\geqmain,\grmain,\rpomain\}$ then $s\ R\ u$.
\end{lemma}

\begin{proof}
By induction on the derivation of $t\ R\ u$.

The case when $R$ is $\ismain$ follows from transitivity of $\ismain$.

The case when $R$ is $\geqmain$ follows from the cases for $\ismain$ and $\grmain$.

The case when $R$ is $\grmain$ follows by a case analysis on the rule used to derive $t \grmain u$.
\begin{itemize}
\item if $t \grmain u$ by Gr-mono, then $s = x\ s_1 \cdots s_n$ and $t = x\ t_1 \cdots t_n$ and
  $u = x\ u_1 \cdots u_n$ and we complete by the induction hypothesis on $s_j \ismain t_j \geqmain
  u_j$ and $s_i \ismain t_i \grmain u_i$;
\item if $t \grmain u$ by Gr-args, then $s = \afun\ s_1 \cdots s_n$ and $\bfun = x\ t_1 \cdots t_n$
  and $u = \cfun\ u_1 \cdots u_n$ with $\afun \eqpred \bfun \eqpred \cfun$ and $\filter(\afun) =
  \filter(\bfun) = \filter(\cfun)$ and we complete by the induction hypothesis on $s_j \ismain t_j
  \geqmain u_j$ and $s_i \ismain t_i \grmain u_i$ (for unfiltered arguments);
\item if $t \grmain u$ by Gr-rpo, we have $s \rpomain u$ and therefore $s \grmain u$ by the
  induction hypothesis.
\end{itemize}

The case when $R$ is $\rpomain$ follows by a case analysis on the rule used to derive $t \rpomain u$.
Note that in this case, $s = \afun\ s_1 \cdots s_n$ and $t = \bfun\ t_1 \cdots t_n$ with
$\afun,\bfun$ having the same type structure $\atype_1 \arrtype \dots \arrtype \atype_m \arrtype
\asort$, $\afun \eqpred \bfun$, $\filter(\afun) = \filter(\bfun) \supseteq \{n+1,\dots,m\}$ and
$s_i \ismain t_i$ whenever $i \in \pi(\afun)$.
\begin{itemize}
\item Rpo-select: $t_i \geqmain u$ for some $i \in \pi(\bfun) = \pi(\afun)$, so
  $s_i \ismain t_i \geqmain u$; we complete with Rpo-select and the induction hypothesis
\item Rpo-appl: $u = u_0\ u_1 \cdots t_m$ and $t \rpomain u_i$ for all $i$, so by the induction
  hypothesis also $s \rpomain u_i$ for all $i$
\item Rpo-copy: $u = \cfun\ u_1 \cdots u_m$ with $\bfun \grpred \cfun$ and $t \rpomain u_i$ for all
  $u_i$ with $i \in \filter(\cfun) \cap \{1,\dots,m\}$; but then also $s \rpomain u_i$ for these
  $i$, and $\afun \eqpred \bfun \grpred \cfun$ implies $\afun \grpred \cfun$
\item Rpo-lex: $u = \cfun\ u_1 \cdots u_m$ with $\bfun \eqpred \cfun$ (so also $\afun \eqpred
  \cfun$ because $\eqpred$ is transitive) and there exists $i$ such that:
  \begin{itemize}
  \item $i \in \filter(\bfun) \cap \{1,\dots,\min(n,m)\} = \filter(\afun) \cap \{1,\dots,\min(n,m)\}$
  \item $s_j \ismain t_j \ismain u_j$ for $j \in \{1,\dots,i-1\} \cap \filter(\afun)$ (as
    $\filter(\afun) = \filter(\bfun)$), so $s_j \ismain u_j$ by transitivity of $\ismain$
  \item $s_i \ismain t_i \grmain u_i$ so by IH $s_i \grmain u_i$,
    since $i \notin \filter(\bfun) = \filter(\afun)$
  \item $s \ismain t \grmain u_j$ for $j \in \{i+1,\dots,m\} \cap \filter(\cfun)$, so also
    $s \grmain u_j$ by the induction hypothesis
    \qedhere
  \end{itemize}
\end{itemize}
\end{proof}

\begin{corollary}\label{cor:compatibility}
If $s \geqmain t \grmain u$ then $s \grmain^+ u$.
\end{corollary}

\begin{lemma}\label{lem:reflexive}
$\geqmain$ is reflexive.
\end{lemma}

\begin{proof}
This immediately follows from reflexivity of $\ismain$ (Lemma~\ref{lem:ismain}).
\end{proof}

\begin{lemma}\label{lem:leftmono}
If $s \grmain t$ with $s :: \atype \arrtype \btype$ and $u \geqmain v$ with $u :: \atype$, then
$s\ u \grmain t\ v$.
\end{lemma}

\begin{proof}
If $s \grmain t$ by Gr-mono or Gr-args, then $s\ u \grmain t\ v$ by Gr-mono or Gr-args as well
(note that in Gr-args, $u$ is filtered away if and only if $v$ is).  So consider the case
$s \grmain t$ by Gr-rpo; that is, $s =  \afun\ s_1 \cdots s_n$ and $s \rpomain t$, which implies
$n+1 \in \filter(\afun)$.  If we can see that $\afun\ s_1 \cdots s_n\ u \rpomain t$ as
well we have $s\ u \rpomain t\ v'$ by Rpo-appl (since $s\ u \rpomain v$ by Rpo-select, because
$n+1 \in \filter(\afun)$ and $u \geqmain v$), and therefore $s\ u \grmain t\ v$ by Gr-rpo.

It remains to prove that if $\afun\ s_1 \cdots s_n \rpomain t$ then $\afun\ s_1 \cdots s_n\ u
\rpomain t$.  We prove that this holds by induction on the derivation of $\afun\ s_1 \cdots s_n
\rpomain t$.  Yet, all cases are simple with the induction hypothesis; we only observe that for
the Rpo-lex case, if $i \in \filter(\afun) \cap \filter(\bfun) \cap \{1,\dots,\min(n,m)\}$ then
also $i \in \filter(\afun) \cap \filter(\bfun) \cap \{1,\dots,\min(n+1,m)\}$, so we can use the
same index.
\end{proof}

\begin{lemma}\label{lem:monotonic}
$\geqmain$ is monotonic.
\end{lemma}

\begin{proof}
Suppose $s :: \atype \arrtype \btype$ and $t :: \atype$, and $s \geqmain s'$ and $t \geqmain t'$.
We must see that $s\ t \geqmain s'\ t'$.

First, if $s \ismain s'$ and $t \ismain t'$, we are done by monotonicity of $\ismain$.

Second, if $s \ismain s'$ and $t \grmain t'$, there are two options:
\begin{itemize}
\item $s = x\ s_1 \cdots s_n$ and $s' = x\ s_1' \cdots s_n'$ with each $s_i \ismain s_i'$ and
  therefore $s_i \geqmain s_i'$; as both $t \geqmain t'$ and $t \grmain t'$ we have $s\ t \grmain
  s'\ t'$ by Gr-mono;
\item $s = \afun\ s_1 \cdots s_n$ and $s' = \bfun\ s_1' \cdots s_n'$ with $\afun \approx \bfun$ and
  $\filter(\afun) = \filter(\bfun$) and $s_i \ismain s_i'$ for $i \in \filter(\afun) \cap \{1,\dots,
  n\}$;
  so if $n+1 \in \filter(\afun) = \filter(\bfun)$ then $s\ t = \afun\ s_1 \cdots s_n\ t \grmain
  \bfun\ s_1' \cdots s_n'\ t' = s'\ t'$ by Gr-args;
  and if $n+1 \notin \filter(\afun) = \filter(\bfun)$ then $s\ t \ismain s'\ t'$ by Eq-args.
\end{itemize}

Finally, if $s \grmain s'$ then $s\ t \grmain s'\ t'$ by Lemma!\ref{lem:leftmono}.
\end{proof}

It remains to prove that $\grmain$ is well-founded.  For this, we use the notion of
\emph{computability}.

\begin{definition}
A term $s$ is terminating if there is no infinite sequence $s \grmain s_1 \grmain s_2 \grmain \dots$

A term $s :: \atype_1 \arrtype \dots \arrtype \atype_m \arrtype \asort$ is computable if for all
computable $t_1 :: \atype_1,\dots,t_m :: \atype_m$, the term $s\ t_1 \cdots t_m$ is terminating.
(This is well-defined by induction on types.)
\end{definition}

\begin{lemma}\label{lem:computability:basics} We observe:
\begin{enumerate}
\item\label{lem:computability:basics:vars}
  All variables are computable.
\item\label{lem:computability:basics:compterm}
  If $s$ is computable, then it is terminating.
\item\label{lem:computability:basics:decrease}
  If $s$ is computable and $s \geqmain t$, then also $t$ is computable.
\end{enumerate}
\end{lemma}

\begin{proof}
By mutual induction on types.

\textbf{Base types.}

(\ref{lem:computability:basics:vars})
Clearly, base-type variables are computable, since there is no $t$ such that $x \grmain t$.

(\ref{lem:computability:basics:compterm})
Base-type computable terms are terminating by definition.

(\ref{lem:computability:basics:decrease})
If $s$ has base type and $s \geqmain t$, then $t$ is terminating: if $t \grmain t_1 \grmain
\dots$ then either $s \grmain t_1 \grmain \dots$ (if $s \ismain t$) or $s \grmain t \grmain
t_1 \grmain \dots$ (if $s \grmain t$), obtaining a contradiction.

\textbf{Arrow type $\atype$: $\atype_1 \arrtype \dots \arrtype \atype_m \arrtype \asort$ with
$m > 0$.}

(\ref{lem:computability:basics:vars})
Let $\avar :: \atype$, and let $s_1 :: \atype_1,\dots,s_m :: \atype_m$ be computable terms.
By induction hypothesis (\ref{lem:computability:basics:compterm}), we can do induction on
$(s_1,\dots,s_m)$ ordered with the place-wise extension of $(\ismain,\grmain)$ to show that
$\avar\ s_1 \cdots s_m$ is terminating.  The only applicable rule is Gr-mono, so if $\avar\ 
s_1 \cdots s_m \grmain t$ then $t = \avar\ t_1 \cdots t_m$ with each $s_i \geqmain t_i$ and
some $s_i \grmain t_i$.  By induction hypothesis (\ref{lem:computability:basics:decrease}),
each $t_i$ is computable, and the tuple is smaller so we complete.

(\ref{lem:computability:basics:compterm})
Let $s:: \atype$ and assume towards a contradiction that $s$ is terminating.
Let $x_1 :: \atype_1,\dots,x_m ::\atype_m$ be variables; by induction hypothesis
(\ref{lem:computability:basics:vars}) they are computable, so $s\ x_1 \cdots x_m$ is terminating.
But Lemma~\ref{lem:leftmono} implies that if $s \grmain s_1 \grmain s_2 \grmain \dots$ then also
$s\ \vec{x} \grmain s_1\ \vec{x}\ \grmain s_2\ \vec{x} \grmain \dots$, contradicting the termination
assumption.

(\ref{lem:computability:basics:decrease})
If $s \grmain t$ then $t$ has the same type structure $\atype$ as $s$.  Let
$u_1 :: \atype_1,\dots,u_n :: \atype_n$ be computable; we must show that
$t\ u_1 \cdots u_n$ is terminating.  But suppose not: $t\ u_1 \cdots u_n \grmain v_1 \grmain v_2
\grmain \dots$.  Then either $s\ u_1 \cdots u_n \grmain v_1 \grmain v_2 \grmain \dots$ (if $s
\ismain t$) by monotonicity of $\ismain$, or $s\ u_1 \cdots u_n \grmain t\ u_1 \cdots u_n
\grmain v_1 \grmain v_2 \grmain\dots$ (if $s \grmain t$) by Lemma \ref{lem:leftmono}.  Hence,
$s\ u_1 \cdots u_n$ is non-terminating, contradicting computability of $s$.
\end{proof}

\begin{lemma}\label{lem:wellfounded:induction}
Let $\afun :: \atype_1 \arrtype \dots \arrtype \atype_m \arrtype \asort$ and terms $s_1 :: \atype_1,
\dots,s_n : \atype_n$ be given such that
(1) any function symbol $\bfun$ with $\afun \grpred \bfun$ is computable;
(2) for all $i \in \filter(\afun)$: $s_i$ is computable, and
(3) if $\afun \eqpred \bfun$ and $t_1,\dots,t_k$ are terms such that $[s_i \mid i \in 
  \filter(\afun)]\ (\ismain,\grmain)_{\mathtt{lex}}\ [t_i \mid i \in \filter(\bfun)]$,
  then $\bfun\ t_1 \cdots t_k$ is terminating.
Then for any $t$ such that $\afun\ s_1 \cdots s_n \rpomain t$: $t$ is computable.
\end{lemma}

\begin{proof}
By induction on the derivation of $\afun\ s_1 \cdots s_n \rpomain t$.

If this follows by Rpo-select, then $s_i \geqmain t$ for some $i \in \filter(\afun)$, so $t$ is
computable by (2) and Lemma~\ref{lem:computability:basics}(\ref{lem:computability:basics:decrease}).

If it follows by Rpo-appl, then $t = t_0\ t_1 \cdots t_k$ and by the induction hypothesis each
$t_i$ is computable, and therefore by definition so is their application.

If it follows by Rpo-copy, then $t = \bfun\ t_1 \cdots t_k$ with $\afun \grpred \bfun$, and by the
induction hypothesis each $t_i$ with $i \in \filter(\bfun)$ is computable; since $\bfun$ is
computable, the whole result is.

If it follows by Rpo-lex, then $t = \bfun\ t_1 \cdots t_k :: \btype_{k+1} \arrtype \dots \arrtype
\btype_p \arrtype \bsort$ with $\afun \eqpred \bfun$.  We must see that for any computable $t_{k+1}
:: \btype_{k+1},\dots,t_p :: \btype_p$ we have $\bfun\ t_1 \cdots t_p$ terminates.  But this is the
case by assumption, because $[s_i \mid i \in \filter(\afun)]\ (\ismain,\grmain)_{\mathtt{lex}}\ 
[t_i \mid i \in \filter(\bfun)]$: we see that this holds because, for some index $i$:
\begin{itemize}
\item $\filter(\afun) \cap \{1,\dots,i\} = \filter(\bfun) \cap \{1,\dots,i\}$, and
\item $s_j \eqpred t_j$ for $j \in \filter(\afun) \cap \{1,\dots,j\}$ \\
  So: $[s_j \mid j \in \filter(\afun) \wedge j \in \{1,\dots,i-1\}] \eqpred_{\text{placewise}}
  [t_j \mid j \in \filter(\afun) \wedge j \in \{1,\dots,i-1\}]$.
\item $i \in \filter(\afun) = \filter(\bfun)$ and $s_i \grmain t_i$ \\
  Hence: any extension of $[s_j \mid j \in \filter(\afun) \wedge j \in \{1,\dots,i\}]\ 
  (\ismain,\grmain)_{\mathtt{lex}}$ any extension of
  $[t_j \mid j \in \filter(\afun) \wedge j \in \{1,\dots,i\}]$.
  \qedhere
\end{itemize}
\end{proof}

\begin{lemma}\label{lem:functionscomputable}
Every function symbol is computable.
\end{lemma}

\begin{proof}
We will prove that for all $\afun :: \atype_1 \arrtype \dots \arrtype \atype_m \arrtype \asort$,
all terms $s_1 :: \atype_1,\dots,s_m :: \atype_n$ such that $s_i$ is computable for $i \in
\{1,\dots,n\} \cap \filter(\afun)$, that $\afun\ s_1 \cdots s_m$ is terminating.
We prove this by induction first on $\afun$, ordered with $\grpred$, and second on
$[s_i \mid i \in \{1,\dots,n\} \cap \filter(\afun)]$, ordered with
$(\ismain,\grmain)_{\mathtt{lex}}$.
To see that the latter is indeed a terminating relation, note that by
Lemma~\ref{lem:computability:basics}, $\grmain$ is terminating on the set of all computable terms
(and this set is closed under $\grmain$), and that the lexicographic extension of a terminating
relation is terminating when the length of sequences is bounded (as we assumed is the case: the
maximal arity of symbols $\eqpred$-equivalent to $\afun$ is bounded).

To see that $\afun\ s_1 \cdots s_n$ is terminating, it suffices to show that $t$ is terminating if
$\afun\ s_1 \cdots s_n \grmain t$.  But there are only two cases to consider: if $\afun\ s_1 \cdots
s_n \grmain t$ by Gr-mono, then $t$ is terminating by the second induction hypothesis because at
least one argument in $\filter(\afun)$ is decreased; and if $\afun\ s_1 \cdots s_n \grmain t$ by
Gr-rpo, then $t$ is computable, and therefore terminating, by Lemma~\ref{lem:wellfounded:induction}.
\end{proof}

\begin{corollary}
$\grmain$ is a well-founded relation on the set of terms.
\end{corollary}

\section{Constrained HORPO for LCSTRSs}\label{sec:constrained}

Having defined our covering pair $(\geqmain,\grmain)$, we are now ready to define a constrained
reduction pair.

\subsection{Ingredients}

Given an LCSTRS, we assume given the following ingredients:
\begin{itemize}
\item a \emph{precedence} $\geqpred$ on the \emph{non-value symbols} (defined the same as in
  Section~\ref{sec:unconstrained}, and with the same restriction that the maximal arity of
  equivalent symbols should be bounded);
\item a \emph{filter} $\filter$, which is restricted so that $\filter(\afun) = \{1,\dots,m\}$ for
  $\afun :: \atype_1 \arrtype \dots \arrtype \atype_m \arrtype \asort$ a theory symbol
\item for every sort $\asort$, a \emph{well-founded ordering} $\grsort_{\asort}$ and a
  \emph{quasi-ordering} $\geqsort_{\asort}$ on values of this sort, such that for all $v_1,v_2$ of
  sort $\asort$: if $v_1 \geqsort_\asort v_2$, then $v_1 \grsort_\asort v_2$ or $v_1 = v_2$ or
  $v_1$ and $v_2$ are minimal wrt $\grsort_{\asort}$
\end{itemize}
We denote $v_1 \grsort v_2$ if $v_1 \grsort_\asort v_2$ for some sort $\asort$ (and similar for
$\geqsort$).

\subsection{Relations}

We provide three mutually recursive relations.

For terms $s,t$ of the same type structure, constraint $\varphi$ and set $\gtvA$ of variables, we
say $s \geqth t$ if this can be obtained by one of the following clauses:
\begin{description}
\item[$\succeq$Theory] $s,t$ are theory terms whose type is the same sort,
  $\Var(s) \cup \Var(t) \subseteq \gtvA$,
  and $\varphi \Vdash s \geqsort t$
\item[$\succeq$Eq] $s\downarrowcalc = t\downarrowcalc$
\item[$\succeq$Mono] $s$ is not a theory term,
  $s = \avar\ s_1 \cdots s_n$ and $t = \avar\ s_1 \cdots s_n$ and each $s_i \geqth t_i$
\item[$\succeq$Args] $s$ is not a theory term,
  $s = \afun\ s_1 \cdots s_n$ and $t = \bfun\ t_1 \cdots t_n$ for $\afun,\bfun$
  function symbols with the same type structure such that $\afun \eqpred \bfun$,
  $\filter(\afun) = \filter(\bfun)$, and $s_i \geqth t_i$ for all $i \in \filter(\afun) \cap
  \{1,\dots,n\}$
\item[$\succeq$Greater] $s \grth t$
\end{description}

For terms $s,t$ of the same type structure, constraint $\varphi$ and set $\gtvA$ of variables, we
say $s \grth t$ if this can be obtained by one of the following clauses:
\begin{description}
\item[$\succ$Theory] $s,t$ are theory terms whose type is the same sort
  $\Var(s) \cup \Var(t) \subseteq \gtvA$,
  and $\varphi \Vdash s \grsort t$
\item[$\succ$Args] $s$ is not a theory term,
  $s = \afun\ s_1 \cdots s_n$ and $t = \bfun\ t_1 \cdots t_n$ for $\afun,\bfun$
  function symbols with the same type structure such that $\afun \eqpred \bfun$,
  $\filter(\afun) = \filter(\bfun)$, $s_i \geqth t_i$ for all $i \in \filter(\afun) \cap
  \{1,\dots,n\}$, and there is some $i \in \filter(\afun) \cap \{1,\dots,n\}$ with
  $s_i \grth t_i$
\item[$\succ$Rpo] $s \rpoth t$
\end{description}

For terms $s,t$ (which may have different type structures) such that $s$ is not a theory term and
has the form $\afun\ s_1 \cdots s_n$ with $\afun$ a function symbol, constraint $\varphi$ and set
$\gtvA$ of variables, we say $s \rpoth t$ if this can be obtained by one of the following clauses:
\begin{description}
\item[$\rpo$Select] $s_i \grth t$ for some $i \in \filter(\afun) \cap \{1,\dots,n\}$
\item[$\rpo$Appl] $t = t_0\ t_1 \cdots t_m$ and $s \grth t_i$ for $1 \leq i \leq m$
\item[$\rpo$Copy] $t = \bfun\ t_1 \cdots t_m$ with $\afun \grpred \bfun$ and 
    $s \rpoth t_i$ for all $i \in \filter(\bfun) \cap \{1,\dots,m\}$
\item[$\rpo$Lex] $t = \bfun\ t_1 \cdots t_m$ with $\afun \eqpred \bfun$ and there is some
    $i \in \filter(\afun) \cap \filter(\bfun) \cap \{1,\dots,\min(n,m)\}$
    such that all of the following hold:
    \begin{itemize}
    \item $\filter(\afun) \cap \{1,\dots,i\} = \filter(\bfun) \cap \{1,\dots,i\}$
    \item $s_j \geqth t_j$ for $j \in \{1,\dots,i-1\} \cap \filter(\afun)$
    \item $s_i \grth t_i$
    \item $s \rpoth t_j$ for $j \in \{i+1,\dots,m\} \cap \filter(\bfun)$
    \end{itemize}
\item[$\rpo$Th] $t$ is a theory term of base type, with all variables in $\gtvA$
\end{description}

\subsection{Coverage}

To see that $\geqmain$ covers $\succeq$ and $\grmain$ covers $\succ$, we start by extending
$\geqpred$ with values: we let
\begin{itemize}
\item $\afun \geqpred v$ for all non-value symbols $\afun$ and all value symbols $v$
  (but not $v \geqpred \afun$, so we have $\afun \grpred v$);
\item $v_1 \geqpred v_2$ if $v_1 \grsort v_2$ or $v_1 = v_2$ or both $v_1$ and $v_2$ are minimal
  wrt $\grsort$ \\
  (hence $v_1 \grpred v_2$ if and only if $v_1 \grsort v_2$, and $v_1 \eqpred v_2$ if either
  $v_1 = v_2$ or both $v_1$ and $v_2$ are minimal wrt $\grsort$)
\end{itemize}
This precedence is still well-founded, by the combination of the well-foundedness of the original
precedence and well-foundedness of $\grsort$.  Now we see:

\begin{lemma}
Let $\gamma$ be a substitution that respects $\varphi$ and maps all variables in $\gtvA$ to ground
theory terms.  Then:
\begin{itemize}
\item if $s \geqth t$ then $(s\gamma)\downarrowcalc \geqmain (t\gamma)\downarrowcalc$
\item if $s \grth t$ then $(s\gamma)\downarrowcalc \grmain (t\gamma)\downarrowcalc$
\item if $s \rpoth t$ then $(s\gamma)\downarrowcalc \rpomain (t\gamma)\downarrowcalc$
\end{itemize}
\end{lemma}

\begin{proof}
By a mutual induction on the derivation of $s\ R\ t$.  Consider the case:
\begin{itemize}
\item If $s \geqth t$ by $\succeq$Theory, then $v_1 := (s\gamma)\downarrowcalc$ and
  $(t\gamma)\downarrowcalc =: v_2$ are both values, with $v_1 \geqsort v_2$.
  By definition of $\geqsort$ this implies either $v_1 \grsort v_2$ or $v_1 = v_2$ or $v_1$ and
  $v_2$ are both minimal wrt $\grsort$; so by our extension of $\geqpred$ we either have
  $v_1 \grpred v_2$ or $v_1 \eqpred v_2$.  In the former case,
  $v_1 \rpomain v_2$ by (Rpo-copy), so $v_1 \grmain v_2$ by (Gr-rpo), so $v_1 \geqmain v_2$ follows.
  In the latter case, $v_1 \ismain v_2$ by (Eq-args), so also $v_1 \geqmain v_2$.
\item Similarly, if $s \grth t$ by $\succ$Theory, then $v_1 := (s\gamma)\downarrowcalc$ and
  $(t\gamma)\downarrowcalc =: v_2$ are both values, with $v_1 \grsort v_2$ and therefore
  $v_1 \grpred v_2$; we conclude $v_1 \grmain v_2$ using Rpo-copy.
\item If $s \geqth t$ by $\succeq$Eq, then $s\downarrowcalc = t\downarrowcalc$ and therefore
  $(s\gamma)\downarrowcalc = (t\gamma)\downarrowcalc$ (as calculation is confluent); by
  reflexivity of $\ismain$ we have $(s\gamma)\downarrowcalc \ismain (t\gamma)\downarrowcalc$ and
  therefore $(s\gamma)\downarrowcalc \geqmain (t\gamma)\downarrowcalc$.
\item If $s \geqth t$ by $\succeq$Args, then because $s$ is not a theory term,
  $(s\gamma)\downarrowcalc = \afun\ (s_1\gamma\downarrowcalc) \cdots (s_n\gamma\downarrowcalc)$;
  and we either have $(t\gamma)\downarrowcalc = \bfun\ (t_1\gamma\downarrowcalc) \cdots
  (t_n\gamma\downarrowcalc)$ or $(t\gamma\downarrowcalc)$ is a value, the latter case only occurring
  if $s$ and $t$ have base type.

  In the first case, we have that each $(s_i\gamma)\downarrowcalc \geqmain
  (t_i\gamma)\downarrowcalc$ by the induction hypothesis, and we easily conclude
  $(s\gamma)\downarrowcalc \geqmain (t\gamma)\downarrowcalc$ either by rule (Eq-args) or by rule
  (Gr-args), depending on whether some $(s_i\gamma)\downarrowcalc \grmain (t_i\gamma)\downarrowcalc$
  or not.

  In the second case, since $s$ is not a theory term, $\afun$ is not a value, and therefore
  $\afun \grpred v := (t\gamma)\downarrowcalc$.  Since $s$ has base type is not a theory term, we
  can apply (Gr-rpo) and (Rpo-copy) to obtain $(s\gamma)\downarrowcalc \grmain v$.
\item If $s \grth t$ by $\succ$Args, we complete in a very similar way.
\item If $s \geqth t$ by $\succeq$Mono, then $s = x\ s_1 \cdots s_n$ and $t = x\ t_1 \cdots t_n$;
  writing $\gamma(x)\downarrowcalc = a\ u_1 \cdots u_k$ where $a$ may be a variable or a function
  symbol and $k \geq 0$, we thus have $(s\gamma)\downarrowcalc = a\ u_1 \cdots u_k\ (s_1\gamma
  \downarrowcalc) \cdots (s_n\gamma\downarrowcalc)$, since $s$ is not a theory term.  Moreover, we
  either have $(t\gamma)\downarrowcalc = a\ u_1 \cdots u_k\ (t_1\gamma\downarrowcalc) \cdots
  (t_n\gamma\downarrowcalc)$ with each $(s_i\gamma)\downarrowcalc \geqmain (t_i\gamma)
  \downarrowcalc$ by the induction hypothesis, or $(t\gamma)\downarrowcalc$ is a value, in which
  case we observe that $a$ cannot be a value and $s$ must have base type, so we may apply
  (Gr-rpo) with (Rpo-copy) and complete as above.
  In the first case, if each $(s_i\gamma)\downarrowcalc \ismain (t_i\gamma)\downarrowcalc$ we have
  $(s\gamma)\downarrowcalc \ismain (t\gamma)\downarrowcalc$ by (Eq-mono) or (Eq-args) (depending
  on whether $a$ is a variable or function symbols); the same holds if $a$ is a function symbol and
  $(s_i\gamma)\downarrowcalc \ismain (t_i\gamma)\downarrowcalc$ holds only for those $i \in
  \filter(a)$.  Otherwise, if some unfiltered $(s_i\gamma)\downarrowcalc \grmain
  (t_i\gamma)\downarrowcalc$ we instead have
  $(s\gamma)\downarrowcalc \grmain (t\gamma)\downarrowcalc$ by (Gr-mono) or (Gr-args) respectively.
\item If $s \geqth t$ because $s \grth t$, we complete immediately by the induction hypothesis.
\item If $s \grth t$ because $s \rpoth t$, we complete immediately by the induction hypothesis and
  rule (Gr-rpo).
\item If $s \rpoth t$ by $\rpo$Th, then $t\gamma$ is a ground theory term, and therefore
  $(t\gamma)\downarrowcalc$ is a value $v$; since $s$ is not a theory term, $\afun$ is not a value
  and therefore $\afun \grpred v$, so $(s\gamma)\downarrowcalc = (\afun \cdots) \rpomain v$ by
  (Rpo-copy).
\item Finally, in each of the cases $\rpo$Select, $\rpo$Appl, $\rpo$Copy and $\rpo$Lex, we either
  immediately conclude with the induction hypothesis (and a case analysis for a smaller ``minimal
  index'' in the case of $\rpo$Lex), or observe that $(t\gamma)\downarrowcalc$ is a value and
  therefore $(s\gamma)\downarrowcalc \rpoth (t\gamma)\downarrowcalc$ by (Rpo-copy), as we also did
  in the case for $\succeq$Args.
  \qedhere
\end{itemize}
\end{proof}

Hence, we have demonstrated that $(\succeq,\succ)$ is a constrained reduction pair.

\bibliography{references}
\end{document}